\documentclass[10pt,a4paper]{article}

\usepackage{amsmath}
\usepackage{amssymb}
\usepackage{graphicx}
\usepackage{makeidx}
\usepackage{algorithm}
\usepackage{algpseudocode}
\usepackage{epstopdf}
\usepackage{adjustbox}

\usepackage{times}
\usepackage[T1]{fontenc}

\newcommand{\eal}{\end{align}}
\newcommand{\beq}{\begin{equation}}
\newcommand{\eeq}{\end{equation}}
\newcommand{\bea}{\begin{eqnarray}}
\newcommand{\eea}{\end{eqnarray}}
\newcommand{\beas}{\begin{eqnarray*}}
\newcommand{\eeas}{\end{eqnarray*}}
\newcommand{\ba}{\begin{array}}
\newcommand{\ea}{\end{array}}

\newcommand{\bit}{\begin{itemize}}
\newcommand{\eit}{\end{itemize}}
\newcommand{\ben}{\begin{enumerate}}
\newcommand{\een}{\end{enumerate}}

\newcommand{\bR}{{\mathbb R}}
\newcommand{\bZ}{{\mathbb Z}}

\newtheorem{theorem}{Theorem}[section]
\newtheorem{proof}{Proof}[section]
\newtheorem{assumption}{Assumption}

\newtheorem{remark}{Remark}[section]
\newtheorem{lemma}{Lemma}[section]

\allowdisplaybreaks[3]

\begin{document}

\date{}

\title{Adaptive model predictive control for constrained, linear time varying systems}
\author{M. Tanaskovic, L. Fagiano, and V. Gligorovski \thanks{M. Tanaskovic is with Singidunum University, Belgrade, Serbia. Email adresses: mtanaskovic@singidunum.ac.rs. L. Fagiano is with Dipartimento di Elettronica, Informazione e Bioingegneria, Politecnico di Milano, Milano, Italy.  Email adresses: lorenzo.fagiano@polimi.it. V. Gligorovski is with Electrical Engineering Department, University of Belgrade, Belgrade, Serbia. Email adresses: voja95@gmail.com.}}

\maketitle

\section{Introduction}\label{S:intro}
This manuscript contains technical details of recent results developed by the authors on adaptive model predictive control
for constrained linear, time varying systems.

\section{Problem Statement}\label{S:problem}
We consider a discrete-time, linear time varying (LTV), multiple input, multiple output (MIMO) system with $n_u$ inputs and $n_y$ outputs. The system is known to be asymptotically stable, but the exact dynamics and the way they change over time are not known. We denote the vector of control inputs at time step $t\in \bZ$ by $u(t)=[u_1(t),\hdots,u_{n_u}(t)]^T$, where $u_i(t)\in \bR,\,i=1,\hdots,n_u$ are the individual plant inputs and $^T$ stands for the matrix transpose operator. In addition, we denote the vector of plant outputs by $y(t)=[y_1(t),\hdots,y_{n_y}(t)]^T$, where $y_j(t)\in \bR ,\, j=1,\hdots, n_y$ are the individual plant outputs. 
At each time step, the dynamic relation between the inputs and the outputs can be described by a linear model of the following form:
\beq\label{Eq:plant_output}
\begin{aligned}
y_j(t)&= H_{j}^T(t) \varphi(t)+d_j(t),\, j=1,\hdots,n_y,
\end{aligned}
\eeq
where $\varphi(t)\in \bR^{m}$ is a regressor vector with $m$ elements, that evolves over time according to the following linear model:
\beq\label{Eq:regressor_evolution}
\begin{aligned}
\varphi(t+1) = F \varphi(t)+Gu(t),
\end{aligned}
\eeq
where $F \in \bR^{m\times m}$ and $G \in \bR^{m\times n_u}$ are known matrices that depend on the considered model parametrization.
\begin{remark}\label{R:A_and_B}
Equations \eqref{Eq:plant_output}-\eqref{Eq:regressor_evolution} cover a broad range of  linear system parameterizations that are used in practice. For example, when  $n_u=1$ and a Finite Impulse Response (FIR) plant model is used, $F$ and $G$ have the following structure:
\beq\label{Eq:FIR_example}
F=\left[ \ba{ccccc} 0 & 0 & \hdots & 0 & 0\\ 1 & 0 & \hdots & 0 & 0\\\vdots & \vdots & \ddots & \vdots & \vdots\\ 0 & 0 & \hdots & 1 & 0 \ea \right], \,\, G=\left[ \ba{c} 1\\0 \\\vdots \\ 0 \ea \right].
\eeq
For the case  $n_u>1$, $F$ and $G$ can be obtained by block diagonalizing the matrices in \eqref{Eq:FIR_example}. Moreover, suitable $F$ and $G$ matrices can be derived for Laguerre \cite{Laguere}, Kautz \cite{Kautz} or generalized basis functions \cite{orthogonal1} parameterizations.
\end{remark}

\begin{remark}\label{R:Same_m}
Note that the same regressor vector is assumed here for all the plant outputs in order to simplify the notation. All the results can easily be extended to the case when different regressor vectors are used for different outputs.
\end{remark}

In \eqref{Eq:plant_output}, the vector  $d(t)=[d_1(t),\hdots,d_{n_y}(t)]^T$, where $d_j(t)\in \bR, \, j=1,\hdots,n_y$, accounts for exogenous additive disturbances and the effects of unmodeled dynamics on the outputs.

Each of the vectors $H_j(t)\in\bR^{m}$ in \eqref{Eq:plant_output} contains the model parameters that describe the influence of $\varphi$ to the plant output $j$ at time step $t$. Defining the matrix $H(t)\in \bR^{n_y\times m}$ as $H(t) \doteq \left[H_1(t), \hdots, H_{n_y}(t)\right]^T$, the dependence of the plant output on the regressor and the disturbance vectors at time step $t$ can be written as:
\beq\label{Eq:plant_output_matrix}
y(t)=H(t)\varphi(t)+d(t).
\eeq
The measured output available for feedback control is corrupted by noise. In particular, the vector of \emph{measured} plant outputs $\tilde{y}(t)$ is given by:
\beq\label{Eq:measured_output}\nonumber
\tilde{y}(t)=y(t)+v(t),
\eeq
where $v(t)=[v_1(t),\hdots,v_{n_y}(t)]^T$ and $v_j(t),\, j=1,\hdots,n_y$ are the individual measurement noise terms that affect each of the measured plant outputs.

\begin{assumption}\label{A:disturbance}(Prior assumption on disturbance and noise)
$d$ and $v$ are bounded as:
\beq\label{Eq:dist_bound}
\ba{lll}
|d_j(t)|&\leq&\epsilon_{d_j}\\
|v_j(t)|&\leq&\epsilon_{v_j}
\ea,\,\forall t\in\bZ,\,\forall j=1,\hdots,n_y,
\eeq
where $\epsilon_{d_j}$ and $\epsilon_{v_j}$ are positive scalars.\
\end{assumption}

We further introduce two additional assumptions on the  system to be controlled. In particular, we assume that, although the system is time varying and the matrix $H(t)$ may change from one time step to the other, the rate of this change is bounded.
\begin{assumption}\label{A:rate_of_change}(Assumption on the bounds on parameter rate of change)
\beq\label{eq:dtheta_bound}
H(t)-H(t-1)=\Delta H(t) \in \mathcal{D}, \forall t \in \bZ,
\eeq
where
\beq\label{eq:change_bounds}
\mathcal{D}\doteq\left\{\Delta H \in \bR^{n_y\times m}: K_j\Delta H_j \leq l_j, j=1,\hdots, n_y\right \},
\eeq
and $K_j\in \bR^{n_{\Delta_j}\times m}$ and $l_j \in \bR^{n_{\Delta_j}},j=1,\hdots,n_y$ are known matrices and vectors that each define a number $n_{\Delta_j}$ of linear inequalities forming nonempty, closed and convex sets, i.e. polytopes.
\end{assumption}
Moreover, we assume that there exists a closed and convex set that is guaranteed to contain the time varying plant parameters at all times.
\begin{assumption}\label{A:system_class}(Assumption on the bounds on parameter values)

The plant model parameters belong to the following parameter set at all times: $H(t)\in \Omega, \forall t\in \bZ$, with
\beq\label{Eq:system_set}
\Omega\! \doteq\! \left\{\!H\!\in \!\bR^{n_y\times m}\!:\! A_{j0}H_j\leq b_{j0}, j=1,\hdots, n_y \!\right\}\!,
\eeq
where the inequalities in \eqref{Eq:system_set} should be interpreted as element-wise inequalities and each matrix $A_{j0}\in \bR^{r_{j0}\times m}$ and vector $b_{j0} \in \bR^{r_{j0}}$ define a nonempty, closed and convex set, i.e. a polytope with $r_{j0}$ faces.
\end{assumption}
\begin{remark}\label{R:justfication2}
Note that Assumptions \ref{A:rate_of_change} and \ref{A:system_class} are not restrictive in practice. In fact, although the system dynamics are generally unknown, the physical principles of operation for any stable system define bounds on the possible values of model parameters. These bounds may be used to define the set $\Omega$ in \eqref{Eq:system_set}. For an example on how to construct such a set for a realistic problem of building climate control, interested reader is referred to \cite{IFAC_building}. Moreover, for any adaptive control scheme to be applicable in practice, the change of the system dynamics must occur with time constants larger than those of the input-output system behavior. Therefore, it is reasonable to assume existence of bounds on the rate of change of the system dynamics. 
\end{remark} 
The control objective is to track a given output reference and reject disturbances over a possibly very long time horizon  $T$ ($T\gg m$), while enforcing input and output constraints:
\begin{subequations}\label{Eq:problem_cost}
\begin{eqnarray}
&
\begin{aligned}
\min\limits_{u(0),\hdots,u(T)}\sum_{t=0}^{T} &\left( y(t)-y_{\text{des}}(t) \right)^TQ \left( y(t)-y_{\text{des}}(t) \right)\\
                                                              &+u(t)^T S u(t) + \Delta u(t)^T R \Delta u(t) \end{aligned}\\\label{Eq:IO_constraints}
&\begin{aligned}
& \text{Subject to},\,\forall t\in[0,T] \\
&\ba{lllll}
C_uu(t)&\leq &g_u\\
C_{\Delta u} \Delta u(t)&\leq&g_{\Delta u}\\
C_y y(t)&\leq& g_y
\ea
\end{aligned}
\end{eqnarray} 
\end{subequations}
where $y_{\text{des}}(t)\in \bR^{n_y}$ is the desired output reference, $Q \in \bR^{n_y\times n_y}$, $S\in \bR^{n_u\times n_u}$ and $R \in \bR^{n_u \times n_u}$ are positive semi-definite weighting matrices selected by the control designer, and $\Delta u(t)=u(t)-u(t-1)$ is the rate of change of the control input. The element-wise inequalities in \eqref{Eq:IO_constraints} define convex sets through the matrices $C_u\in\bR^{n_i\times n_u}$, $C_{\Delta u}\in\bR^{n_{\Delta u}\times n_u}$, $C_y\in\bR^{n_o\times n_y}$ and the vectors $g_u\in\bR^{n_i}$, $g_{\Delta u}\in\bR^{n_{\Delta u}}$, $g_y\in\bR^{n_o}$, where $n_u$, $n_{\Delta u}$ and $n_y$ are the number of linear constraints on the inputs, input rates, and outputs, respectively. We assume that the set defining the constraints on $\Delta u(t)$ contains the origin and that the constraint set of $u(t)$ is compact, which are assumptions that are satisfied in most practical problems.

\section{Adaptive control algorithm}\label{S:Overall_algorithm}
The optimization problem \eqref{Eq:problem_cost} is generally intractable. As a feasible approximate solution, we propose the use of a receding horizon control policy that relies on two steps: 1) a recursive set membership identification that tracks the set of all possible model parameters (feasible parameter set) consistent with initial assumptions and data, and 2) a  model predictive controller that exploits the model set to robustly enforce constraints while optimizing the plant behavior. The approach is outlined in Algorithm \ref{A: Main_algorithm}.

\begin{algorithm}\label{alg:controller}
	At time step $k$:
	\begin{itemize}
		\item[1)]Compute the current feasible parameter set by taking into account the latest output measurement and considering the worst case parameter change. Calculate a nominal model of the plant based on the updated feasible parameter set;
		\item[2)] Compute an optimal input sequence that minimizes a cost function with respect to the nominal model, and guarantees robust satisfaction of constraints for all parameters inside the feasible parameter set, also taking into account the possible future parameter changes;
		\item[3)] Apply the first input from the sequence, set $k=k+1$ and go to 1).
	\caption{Proposed adaptive receding horizon control algorithm.}\label{A: Main_algorithm}
\end{itemize}
\end{algorithm}

We now describe in detail these two main steps.

\subsection{Recursive set membership identification algorithm}\label{S:SM_ID}

The proposed recursive set membership identification algorithm is based on the fact that, due to Assumption \ref{A:disturbance}, for each of the plant outputs, at any given time step $t$, the absolute difference between the output  measurement and the output prediction based on the plant model can not be larger then the sum of the corresponding disturbance and noise bounds. Therefore, each new measurement collected from the plant at time step $t$, defines a set to which the parameter matrix $H(t)$ is guaranteed to belong to at time step $t$:
\beq\label{eq:slabs_t}
\mathcal{S}_t(t) \doteq \left\{ \ba{ll} H \in \bR^{n_y\times m}:& \left|H_j^T\varphi (t) - \tilde{y}_j(t)\right|\leq \epsilon_{d_j}+\epsilon_{v_j}, \\ & j=1,\hdots, n_y\ea\right\} 
\eeq
where $\mathcal{S}_i(j)$ denotes the set that is defined by the regressor and output measurement vectors at time step $i$, i.e. $\varphi(i)$ and $\tilde{y}(i)$, and that is guaranteed to contain the model parameter matrix $H(j)$ at time step $j$. In particular, the set $\mathcal{S}_t(t)$ is formed by $n_y$ slabs that are defined by the regressor vector $\varphi(t)$ and the output measurements $\tilde{y}_j(t),j=1,\hdots,n_y$ collected at time step $t$.

In addition, we note that the relation between the model parameter matrix at time step $t$, $H(t)$, and the regressor and plant output vectors at time step $t-1$, i.e. $\varphi(t-1)$ and $y(t-1)$, can be expressed by the following equation:
\beq\label{eq:mixed_dynamics}
y(t-1)=H(t)\varphi(t-1)+d(t-1)+\vartheta(t-1),
\eeq
where $\vartheta(t-1)\in \bR^{n_y}$, $\vartheta(t-1)=[\vartheta_1(t-1),\hdots,\vartheta_{n_y}(t-1)]^T$, and $\vartheta_j(t-1)\in \bR,j=1,\hdots,n_y$ are the contributions of the unmodeled dynamics to the individual plant outputs, present due to the fact that the parameter matrix $H(t)$ is used insterad of the matrix $H(t-1)$ in order to relate the regressor vector $\varphi(t-1)$ and the output vector $y(t-1)$:
\beq
\vartheta(t-1)\doteq\left(H(t-1)-H(t)\right)\varphi(t-1).
\eeq

From Assumption \ref{A:rate_of_change}, it follows that the signal $\vartheta(t-1)$ is bounded such that it holds:
\beq
\underline{\vartheta}_j(t-1)\leq\vartheta_j(t-1)\leq\overline{\vartheta}_j(t-1), j=1,\hdots, n_y,
\eeq

where each of the bounds $\underline{\vartheta}_j(t-1)\in \bR$ and $\overline{\vartheta}_j(t-1)\in \bR$, $j=1,\hdots, n_y$ is given as the solution of the following two linear programs (LPs):
\beq\label{eq: time_varying_bounds}
\ba{ll}
\underline{\vartheta}_j(t-1) \doteq & \min\limits_{x \in \bR^m}\varphi^T(t-1)x\\
\overline{\vartheta}_j(t-1) \doteq & \max\limits_{x}\varphi^T(t-1)x\\
&\text{Subject to:}\\
& K_jx\leq l_j.
\ea
\eeq

Based on these definitions, we define the set $\mathcal{S}_{t-1}(t)$ as the set that is formed on the basis of the regressor and the output measurement vectors at time step $t-1$, i.e. $\varphi(t-1)$ and $\tilde{y}(t-1)$, and is guaranteed to contain the matrix of model parameters at time step $t$, i.e. $H(t)$:
\beq\label{eq: s_tm1_t}
\mathcal{S}_{t\!-\!1}(t)\! \!\doteq\!\!\left\{\!\! \ba{l} H\in \bR^{n_y\times m}: \\\!-\!\epsilon_{d_j}\!-\!\epsilon_{v_j}\!+\!\underline{\vartheta}_j(t\!-\!1) \leq H_j^T\varphi(t\!-\!1)\!-\!\tilde{y}_j(t\!-\!1), \\
 H_j^T\varphi(t\!-\!1)\!-\!\tilde{y}_j(t\!-\!1)\leq \epsilon_{d_j} \!+\! \epsilon_{v_j}\!+\!\overline{\vartheta}_j(t\!-\!1),\\
 j=1,\hdots n_y
\ea\!\!\!
\right\}\!\!.
\eeq
More generally, following the same logic, we may define the set $\mathcal{S}_{k}(t)$ as the set formed on the basis of the regressor and output measurement vectors at time step $k\leq t$, i.e. $\varphi(k)$ and $\tilde{y}(k)$, that is guaranteed to contain the matrix of model parameters at time step $t$, $H(t)$ as:
\beq\label{eq: s_k_t}
\mathcal{S}_{k}(t)\! \doteq\!\left\{\!\! \ba{l} H\in \bR^{n_y\times m}: \\\!-\!\epsilon_{d_j}\!-\!\epsilon_{v_j}\!+\!(t\!-\!k)\underline{\vartheta}_j(k) \leq H_j^T\varphi(k)\!-\!\tilde{y}_j(k), \\
 H_j^T\varphi(k)\!-\!\tilde{y}_j(k)\leq \epsilon_{d_j} \!+\! \epsilon_{v_j}\!+\!(t\!-\!k)\overline{\vartheta}_j(k),\\
 j=1,\hdots n_y
\ea\!
\right\}\!.
\eeq

Based on the definition of the set $\mathcal{S}_{k}(t)$ in \eqref{eq: s_k_t} and the Assumptions \ref{A:disturbance}, \ref{A:rate_of_change} and \ref{A:system_class}, we define the feasible parameter set at time step $t$, denoted by $\mathcal{F}(t)$, as the set that is guaranteed to contain all model parameter matrices at time step $t$, i.e. $H(t)$, that are consistent with the initial assumptions and the output measurements collected up to time step $t$. The feasible parameter set is given by the intersection of the set $\Omega$ and all the sets $\mathcal{S}_{k}(t), k=1,\hdots, t$:
\beq\label{eq: FPS}
\mathcal{F}(t) \doteq \Omega \cap \left(\bigcap\limits_{k=1,\hdots, t} \mathcal{S}_k(t) \right).
\eeq
According to Assumption \ref{A:system_class}, the set $\Omega$ is defined through polytopic constraints on the rows of the parameter matrix $H(t)$. Moreover, the sets $\mathcal{S}_k(t)$, $k=1,\hdots, t$ are defined through linear inequality constraints on the rows of the matrix $H(t)$, defined by the measured data. Therefore, the feasible parameter set $\mathcal{F}(t)$ is also given by polytopic constraints on the rows of the model parameter matrix $H(t)$. This means that $\mathcal{F}(t)$ can be uniquely described by a set of matrices and vectors that define the polytopic constraints on each of the rows of matrix $H(t)$:
\beq\label{eq: FPS_polytope}
\mathcal{F}(t) = \left\{ H \in \bR^{n_y \times m}:A_j(t)H_j\leq b_j(t)\right\},
\eeq
where each of the matrices and vectors $A_j(t)\in \bR^{r_j(t)\times m}, b_j(t) \in \bR^{r_j(t)}$, $j=1,\hdots, n_y$ define $r_j(t)$ linear inequalities. 

In order to use the defined feasible parameter set $\mathcal{F}(t)$ to compute the control inputs on-line, a recursive update approach is needed. To this end, we note that the matrix $A_j(t)$ can be created from the matrix $A_j(t-1), j=1,\hdots, n_y$ by appending two rows formed by the regressor vector at time step $t$, $\varphi(t)$ and that the vector $b_j(t)$ can be formed from the vector $b_j(t-1),j=1,\hdots, n_y$, by first adding the terms that should account for the possible change of the plant model with respect to the previous time step and then by appending two new rows that define the constraints related to the newly collected output measurement $\tilde{y}_j(t), j=1,\hdots, n_y$:
\beq\label{eq: matrix_vector_update}
A_j(t)\!=\!\left[ \ba{c} A_j(t\!-\!1)\\-\varphi^T(t)\\\varphi^T(t)\ea \right], \, b_j(t)\!=\!\left[\ba{c} b_j(t\!-\!1)\!+\!\Delta b_j(t\!-\!1)\\-\tilde{y}_j(t)\!+\!\epsilon_{d_j}\!+\!\epsilon_{v_j}\\
 \tilde{y}_j(t)\!+\!\epsilon_{d_j}\!+\!\epsilon_{v_j}\ea \right]\!,
\eeq  
where the vectors $\Delta b_j(t-1)\in \bR^{r_j(t-1)},j=1,\hdots, n_y$ contain the bounds on the output perturbation induced by all the  possible changes of the model dynamics from one time step to the next:
\beq
\Delta b_j(t-1)=\left[\ba{c}\bold{0}_{r_{jo}}\\-\underline{\vartheta}_j(0)\\\overline{\vartheta}_j(0)\\ \vdots \\-\underline{\vartheta}_j(t-1)\\\overline{\vartheta}_j(t-1)  \ea\right],
\eeq 
with $\bold{0}_{r_{jo}}\in \bR^{r_{j0}}$ denoting a vector of $r_{jo}$ zeros.
 
Using the recursive equation \eqref{eq: matrix_vector_update} to update the matrices  $A_j(t)$ and vectors $b_j(t)$, $j=1,\hdots, n_y$, would result, in general, in a growth of their dimension $r_j(t), \, j=1,\hdots, n_y$ by two with each new output measurement. In this way, keeping track of the matrices $A_j(t)$ and vectors $b_j(t)$ over time would become intractable. Therefore, in order to have a tractable recursive identification algorithm, we keep track of the constraints that were generated by the last $M$ measurements, where $M$ is an even number and a design parameter. In this way the dimensions of the matrices $A_j(t)$ and the vectors $b_j(t)$ remain bounded over time, such that $r_j(t)\leq r_{0j}+M, \forall j=1,\hdots, n_y, \forall t$. The parameter $M$ should be selected such that a good trade-off between conservativeness and computational complexity is reached. Namely, if $M$ is selected too small, the resulting approximation of the feasible parameter set would be conservative. On the other hand, choosing $M$ too large would require a lot of memory and computational power for the implementation of the proposed algorithm. 
\begin{remark}\label{R:selection_ofM}
In the described approach, taking into account the worst-case time variation of the system results in a growth of the uncertainty related to each collected measurement pair $(\varphi,\tilde{y})$  with time. This can be seen in \eqref{eq: s_k_t}, where the width of the hyperslab defined by a given output measurement depends on the difference between the current time step and the time step at which that measurement was taken. Therefore, as time goes on, the inequalities defined by old measurements will become redundant. This makes bounding the complexity of the feasible parameter set by discarding constraints related to past measurements a natural choice. 
\end{remark}

 Based on the described way to recursively update the matrices $A_j(t)$ and vectors $b_j(t)$, $j=1,\hdots, n_y$ and the presented strategy to bound the growth of their dimension, in Algorithm \ref{A: Recusive_identification}, we propose a recursive set membership identification algorithm that can be used to update the feasible parameter set $\mathcal{F}(t)$ at each time step. 

\begin{algorithm}
\begin{itemize}
  \item[1)]At time step $t=0$, for $j=1,\hdots, n_y$, set $A_j(0) = A_{j0}$, $ b_j(0) = b_{j0}$;
\item[2)] At time step $t>0$, calculate the regressor vector $\varphi(t)$ according to \eqref{Eq:regressor_evolution} and take the measurement vector $\tilde{y}(t)$;
\item[3)] For $j=1,\hdots,n_y$, calculate $\underline{\vartheta}_j(t)$ and $\overline{\vartheta}_j(t)$ by solving linear programs as in \eqref{eq: time_varying_bounds};
\item[4)] For $j=1,\hdots, n_y$ form the matrix $A_j(t)$ and the vector $b_j(t)$ from $A_j(t-1)$ and $b_j(t-1)$ according to \eqref{eq: matrix_vector_update};
\item[5)] For $j=1,\hdots,n_y$, if $r_j(t)>r_{j0}+M$, remove the $r_{j0}+1$ and if needed $r_{j0}+2$ row from the matrix $A_j(t)$ and vector $b_j(t)$, such that after removal it holds that $r_j(t)\leq r_{j0}+M$;
  \item[6)] Set $t=t+1$, go to 2).
\end{itemize}
\caption{Recursive algorithm for updating the feasible parameter set}\label{A: Recusive_identification}
\end{algorithm}

Algorithm \ref{A: Recusive_identification} guarantees that under the Assumptions \ref{A:disturbance}, \ref{A:rate_of_change} and \ref{A:system_class}, the actual model parameter matrix $H(t)$ always belongs to  the feasible parameter set $\mathcal{F}(t)$, as formally stated in Lemma \ref{L:nonempty_membership} later on.

In addition to the model set, the proposed identification algorithm provides a nominal model of the plant at each time step (see Algorithm \ref{A: Main_algorithm}). The latter is given by a matrix $H_{c}(t)\in \bR^{n_y\times m}$, $H_c=[H_{c,1},\hdots H_{c,n_y}]^T$, where the vectors  $H_{c,j}(t)\in \bR^{m},\, j=1,\hdots, n_y$ can be calculated by solving an LP that aims to find the point inside the feasible parameter set $\mathcal{F}(t)$ that is closest to the nominal model in the previous time step (i.e. $H_c(t-1)$):
\beq\label{Eq:central_estimate}
\begin{aligned}
&\min \limits_{H_{c,j}(t), j=1,\hdots, n_y} \sum \limits_{j=1}^{n_y}\|H_{c,j}(t-1)-H_{c,j}(t)\|_1\\
&\text{Subject to:}\\
&  A_j(t)H_{cj}(t)\leq b_j(t), \,\,\forall j=1,\hdots, n_y.
\end{aligned}
\eeq
The matrix $H_{c}(0)$ can be initialized as an arbitrary nonzero element inside the set $\Omega$.
\begin{remark}\label{R:feasibility_invalidation}
Note that if the optimization problem \eqref{Eq:central_estimate} has no feasible solution, it means that $\mathcal{F}(t)=\emptyset$, i.e. the collected data invalidate the initial assumptions. This may happen in practice if a sudden and unexpected change in the plant dynamics occurs, which violates Assumption \ref{A:rate_of_change}. In such cases,  the recursive algorithm to update the feasible parameter set could be restarted and $\mathcal{F}(t)$ could be reinitialized with the set $\Omega$ (see e.g. Assumption \ref{A:system_class}). Therefore, the fact that the feasible parameter set $\mathcal{F}(t)$ becomes empty could be used to detect abrupt changes in the system dynamics, and to properly react to such cases in practice. This aspect is interesting in the framework of fault detection techniques.
\end{remark}

\subsection{Finite horizon optimal control problem}\label{S:Control}
Let $u(k|t),\, k\in[t,t+N-1]$, $N\geq m$,  be the candidate future control moves, where the notation $k|t$ indicates the prediction at step $k\geq t$ given the information at the current step $t$. For brevity, we collect these decision variables in vector $U\doteq[u(t|t)^T \ldots u(t+N-1|t)^T]^T$. We also define the vectors of future input increments $\Delta u(k|t), \, k\in [t,t+N-1]$ as:
\beq\nonumber
\!\Delta u(k|t)\!=\!\begin{cases} u(t|t)\!-\!u(t-1)\!\! & \text{if}\, k=t\\
u(k|t)\!-\!u(k\!-\!1|t)\!\! & \text{if}\, t\!+\!1\leq k\leq t\!+\!N\!-\!1. \end{cases}
\eeq
Moreover, we define the future regressor vectors $\varphi(k|t)\in\bR^{m},\,k\in[t+1,t+N]$ as:
\begin{equation}\label{Eq:predicted_regressors}
\varphi(k|t)\!=\!\!\begin{cases} F\varphi(t)\!\!+\!\!Gu(t|t)\!\!\! & \text{if} \,\, k=t+1\\
                                                 F \varphi(k\!-\!1|t)\!\!+\!\!Gu(k\!-\!1|t) \!\!\!& \text{if} \,\, t\!+\!2\!\leq\! k\!\leq\! t\!+\!N.
\end{cases}
\end{equation}

In addition, we define the current prediction error $\hat{d}(t)\in \bR^{n_y}$ as the difference between the measured plant output and the one predicted by the nominal model at time step $t$:
\beq\label{Eq:predicion_error}
\hat{d}(t)\doteq \tilde{y}(t)-H_c(t)\varphi(t).
\eeq
Then, we consider the following cost function:
\beq \label{Eq:cost_function}
\begin{aligned}
&J(U,\tilde{y}(t),\varphi(t))\doteq\\
& \sum\limits_{k=t}^{t+N-1}\!\!\left(\hat{y}(k\!+\!1|t)\!-\!y_{\text{des}}(k\!+\!1|t)\right)^TQ(\hat{y}(k\!+\!1|t)\\&-\!y_{\text{des}}(k\!+\!1|t))\!+\! u(k|t)^TSu(k|t)+\Delta u(k|t)^TR\Delta u(k|t),
\end{aligned}\!\!
\eeq
where:
\beq\label{Eq:predicted_output}
\hat{y}(k+1|t)=H_c(t)\varphi(k+1|t)+\hat{d}(t).
\eeq
In \eqref{Eq:cost_function}, $\tilde{y}(t)$ and $\varphi(t)$ are  known parameters and $y_\text{des}(k|t),\,k\in[t+1,t+N]$, are the predicted values of the desired output. Note that, if the nominal model of the plant $H_c(t)$ were equal to the real plant, which would not change in the considered time horizon, the measurement noise $v(t)$ were zero, and the output disturbance $d(t)$ were constant, for $N=T$, minimizing the cost function \eqref{Eq:cost_function} would be equivalent to minimizing the cost function of the control objective \eqref{Eq:problem_cost}.

Satisfaction of input constraints can be enforced by the following set of inequalities:
\beq\label{Eq:input_rate}
\ba{ll}
\begin{aligned}
C_u u(k|t) & \leq & g_u\\
C_{\Delta u} \Delta u(k|t) & \leq & g_{\Delta u}
\end{aligned} & \,\,\forall k\in[t,t+N-1].
\ea
\eeq

In order to define the output constraints, we first introduce the notion of the predicted feasible parameter set, which we denote by $\mathcal{F}(k|t), \, k \in [t+1,t+N]$. These essentially propagate the feasible parameter set in the future. They are computed as if the recursive identification Algorithm \ref{A: Recusive_identification} were applied at predicted time step, but without taking into account the future output measurements, which are unknown at the current time step. The terminal predicted feasible parameter set, $\mathcal{F}(t+N|t)$, is chosen as equal to the uncertainty set $\Omega$, to which the model parameters are guaranteed to belong to at all times:
\beq\label{eq:predicted_FPS}
\mathcal{F}(k|t) = \left\{ H \in \bR^{n_y \times m}:A_j(k|t)H_j\leq b_j(k|t)\right\},
\eeq
where the predicted matrices $A_j(k|t)$ and the vectors $b_j(k|t)$, for $k\in [t+1,t+N-1]$ and $j = 1, \hdots, n_y$ are given as:
\small
\begin{equation}\label{Eq:predicted_matrices1}
A(k\!+\!1|t)\!=\!\!\begin{cases} A(k|t) &\!\! \text{if} \,\, r_j(k|t)\!\leq\! M'\\ \\
                                                 \left[\!\!\ba{c}a_{j1}(k|t)\\\vdots\\a_{jr_{j0}}(k|t)\\a_{jr_{j0}+3}(k|t)\\\vdots\\a_{jr_{j}(t)}(k|t) \ea\!\! \right]&\!\! \text{otherwise},
\end{cases}
\end{equation}

\begin{equation}\label{Eq:predicted_matrices2}
b(k\!+\!1|t)\!=\!\!\begin{cases} b(k|t)\!+\left[\!\!\ba{c}\bold{0}_{r_{jo}}\\-\underline{\vartheta}_j\left(t\!-\!\frac{r_j(t)\!-\!r_{j0}}{2}\right)\\\overline{\vartheta}_j\left(t\!-\!\frac{r_j(t)\!-\!r_{j0}}{2}\right)\\ \vdots \\-\underline{\vartheta}_j(t)\\\overline{\vartheta}_j(t)  \ea\!\right]\,\,\, \text{if}\,\,\, r_j(k|t)\!\leq\! M'\, &\\ \\
                                                \left[\!\!\ba{c}b_{j1}(k|t)\\\vdots\\b_{jr_{j0}}(k|t)\\b_{jr_{j0}+3}(k|t)-\underline{\vartheta}_j\left(k\!-\!\frac{r_j(t)\!-\!r_{j0}}{2}\right)\\            b_{jr_{j0}+4}(k|t)+\overline{\vartheta}_j\left(k\!-\!\frac{r_j(t)\!-\!r_{j0}}{2}\right) \\\vdots\\ b_{jr_{j}(t)-1}(k|t)-\underline{\vartheta}_j(t) \\b_{jr_{j}(t)}(k|t)+\overline{\vartheta}_j(t) \ea\!\! \right]\,\, \text{otherwise},&
\end{cases}
\end{equation}
\normalsize
where $a_{ji}(k|t)$ and $b_{ji}(k|t)$ denote the $i^{\text{th}}$ row of the matrix $A_{j}(k|t)$ and the vector $b_{j}(k|t)$ respectively, $r_j(k|t)=r_j(t)+2(k-t)$ represents the predicted dimension of the matrices $A_j(k)$ and the vectors $b_j(k)$ that would be obtained by using Algorithm \ref{A: Recusive_identification} if no rows would be removed (i.e. if the dimension of the matrices and vectors would be allowed to grow without limit in the future), and $M' = M+r_{j0}$ is a constant. The initial predicted matrices $A_j(t|t)$ and the vectors $b_j(t|t), j=1,\hdots, n_y$ correspond to their actual values at time step $t$:
\beq
A_j(t|t) = A_j(t),\, \, b_j(t|t) = b_j(t).
\eeq

\begin{remark}\label{R:terminal_set}
Setting $\mathcal{F}(t+N|t)=\Omega$ introduces additional conservativeness, since the set $\mathcal{F}(t+N|t)$ could be calculated from the set $\mathcal{F}(t+N-1|t)$ in the same way as for the sets $\mathcal{F}(k|t), k\in[t+1,t+N-1]$, and in general such a set would be tighter than the set $\Omega$. However, this approach enables recursive feasibility  (see Theorem \ref{L:reqursive_feasibility_tv} later on).
\end{remark}
The robust satisfaction of the output constraints is  guaranteed by enforcing them for all the parameters inside the predicted feasible parameter sets $\mathcal{F}(k|t), k\in[t+1,t+N]$ and for all disturbance realizations:
\beq\label{Eq:output}
C_yH\varphi(k|t)\!+\!\overline{d}\leq g_y, \, \forall H\! \in\! \mathcal{F}(t),\, \forall k\in[t\!+\!1,t\!+\!N],
\eeq
where $\overline{d}=[\overline{d}_1,\hdots, \overline{d}_{n_o}]^T$, and $\overline{d}_l \in \bR,\, l=1,\hdots, n_o$ are given as:
\beq\label{Eq:iter_LP}\nonumber
\overline{d}_l=\sum\limits_{j=1}^{n_y}|c_{lj}|\epsilon_{d_j},
\eeq
where $c_{lj}$ stands for the element of the $l^\text{th}$ row and $j^\text{th}$ column of the matrix $C_y$. 

However, constraints \eqref{Eq:output} can not be used directly, as this would result in an infinite-dimensional bilinear optimization problem that is very hard to solve in general. Constraints  \eqref{Eq:output} can be reformulated into a set of linear equalities and inequalities by introducing additional decision variables and using duality of linear programs. Here, we state the result related to this reformulation without giving the proof, as it is very similar to Lemma 3.2 in \cite{OurAdaptive_MPC}. To this end, we introduce the vector of auxiliary decision variables $\Lambda\doteq\left[\Lambda_1^T,\ldots,\Lambda_{n_o}^T\right]^T\in \bR^{n_oNr(t)}$, where $\Lambda_l\doteq [\lambda_l(t+1|t)^T,\ldots,$$\lambda_l(t+N|t)^T]^T,\,l=1,\ldots,n_o$, and for each $k=t+1,\ldots,t+N$, $\lambda_l(k|t) \in \bR^{r(t)}$ and $r(t)=\sum_{j=1}^{n_y} r_j(t)$.
\begin{lemma}\label{L:constraint_equivalence_1} Lemma 3.2 from \cite{OurAdaptive_MPC}

The constraints (\ref{Eq:output}) are satisfied if and only if there exist $\varphi(k|t)$, $ k\in[t+1,t+N]$ and $\Lambda$ such that the following set of inequalities is feasible:
\small
\beq\label{Eq:relaxed_constraints1}
\left.\ba{l}\!\!
A(k|t)^T\lambda_l(k|t)=\left[\ba{c}
c_{l1}\varphi(k|t)\\
\vdots\\
c_{ln_y}\varphi(k|t)
\ea\right]\\
b(k|t)^T\lambda_l(k|t)\leq o_l\!-\!\overline{d}_l\\
\lambda_l(k|t)\geq \bold 0
\ea \right\}\ba{cc}\forall l=1,\hdots,n_o\\
\forall k\!\in\![t\!+\!1,t\!+\!N]\ea
\eeq
\normalsize
with
\small
\beq\label{Eq:explanation}\nonumber
\begin{aligned}
A(k|t)&=\left[\ba{cccc} A_1(k|t)& \bold 0& \ldots&\bold 0\\
\bold 0& A_2(k|t)& \ldots& \bold 0\\
\vdots& \vdots & \ddots & \vdots \\
\bold 0&\bold 0& \ldots&A_{n_y}(k|t)\ea \right]\\
b(k|t)&=\left[\ba{c}
b_1(k|t)\\
\vdots\\
b_{n_y}(k|t)
\ea \right],
\end{aligned}
\eeq
\normalsize
where $\bold 0$ represents zero matrices of appropriate dimensions and $o_l$ is the $l^\text{th}$ element of the vector $g_y$. 
\end{lemma}

To guarantee recursively feasibility, we introduce an additional generalized terminal equality constraint, as done e.g. in \cite{FaTe12}:
\beq\label{Eq:end_constraint}
\varphi(t+N|t)=F\varphi(t+N|t)+Gu(t+N-1|t).
\eeq
This means that we require the terminal regressor to correspond to a steady state for the considered model structure.

For fixed values of $N$, $Q$, $S$ and $R$, we can now define the finite horizon optimal control problem (FHOCP) to be solved at each time step $t$:
\beq \label{Eq:FHOCP}
\begin{aligned}
\min
\limits_{U,\Lambda} &J(U,\tilde{y}(t),\varphi (t))\\
&\text{Subject to:}\,\,\, \eqref{Eq:input_rate},\, \eqref{Eq:relaxed_constraints1}, \, \eqref{Eq:end_constraint},
\end{aligned}
\eeq
which is a quadratic program (QP), that can be efficiently solved in general. The number of decision variables and constrains of the QP \eqref{Eq:FHOCP} depends on the chosen prediction horizon $N$ and the dimension of matrices and vectors that define the feasible parameter set $\mathcal{F}(t)$. Therefore, the computational complexity of \eqref{Eq:FHOCP} can be decreased by reducing the tuning parameter $M$, which bounds the dimension of matrices $A_j(t)$ and the vectors $b_j(t), j=1,\hdots, n_y$, at the cost of higher conservativeness as discussed in section \ref{S:SM_ID}.\\

\section{Properties of the proposed adaptive control algorithm}\label{S:Properties}

The described control algorithm guarantees recursive feasibility and robust satisfaction of both  input and output constraints. In order to formally state and prove this, we first state two results that are instrumental to prove the main result.

\begin{lemma}\label{L:nonempty_membership}
Let Assumptions \ref{A:disturbance}-\ref{A:system_class} hold. Then a feasible parameter set $\mathcal{F}(t)$ obtained by using the recursive Algorithm \ref{A: Recusive_identification} is a nonempty set that is guaranteed to contain the true model parameter matrix at each time step, i.e. $\mathcal{F}(t)\neq\emptyset$ and $H(t) \in \mathcal{F}(t), \forall t\geq 0$.
\end{lemma}
\begin{proof}
See the Appendix.$\hfill\square$
\end{proof}

\begin{lemma}\label{L:predicted_FPS}
Let Assumptions \ref{A:disturbance}-\ref{A:system_class} hold. Then,  when Algorithm \ref{A: Recusive_identification} is used, at each time step $t$, it holds that $\mathcal{F}(k|t+1)\subseteq \mathcal{F}(k|t),\;k\in[t+2, t+N]$.
\end{lemma}

\begin{proof}
See the Appendix. $\hfill\square$
\end{proof}

We now state the main result related to recursive feasibility of the finite horizon optimal control problem and robust constraint satisfaction.

\begin{theorem}\label{L:reqursive_feasibility_tv}
Let Assumptions \ref{A:disturbance}-\ref{A:system_class} hold, and assume that the problem \eqref{Eq:FHOCP}, solved under the proposed adaptive control scheme that uses the recursive set membership identification Algorithm \ref{A: Recusive_identification}, is feasible at time step $t=0$. Then the problem  \eqref{Eq:FHOCP} is recursively feasible and the closed-loop system obtained by applying the proposed adaptive algorithm is guaranteed to satisfy input and output constraints $\forall t \geq 0$.
\end{theorem}
\begin{proof}
See the Appendix. $\hfill\square$
\end{proof}
\begin{remark}\label{R:temporary_better2}
The two key components that allow us to guarantee recursive feasibility are Assumption \ref{A:rate_of_change} (known bounds on the parameters rate of change) and the robustification of the output constraints at the end of the prediction horizon with respect to the whole set $\Omega$. The theoretical guarantees are therefore achieved by increasing the conservativeness of the overall adaptive MPC algorithm. 
Such conservativeness can be mitigated by increasing the prediction horizon $N$. In this way, the presence of the terminal constraint does not have a large impact on the control performance at the beginning of the prediction horizon. Due to the receding horizon strategy in which the feasible parameter set is updated at each time step, the resulting control performance of the proposed adaptive scheme also remains unaffected by this conservativeness. In fact, the presence of the terminal constraint can be seen as a way for the controller to ensure that it can satisfy the constraints for all possible future changes of the plant parameters and if the plant parameters can not change a lot from one time step to the other and if the prediction horizon is long enough, the effects of the terminal constraint on the controller performance are not significant. 
\end{remark}


The proposed adaptive control algorithm requires the solution of $2n_y+1$ LPs that can be parallelized, and of a single QP at each time step. These convex optimization problems can be solved very efficiently with available software tools. Moreover, since the Algorithm \ref{A: Recusive_identification} for the recursive updating of the set $\mathcal{F}(t)$ uses bounded complexity updating strategy, all the matrices and vectors used for describing the set $\mathcal{F}(t)$ are guaranteed to have bounded dimensions, and hence the size of the LPs and the QP that have to be solved at each time step is limited. All of these properties make the proposed adaptive control algorithm computationally tractable and suitable for on-line implementation.

\section{Simulation study}\label{S:Example}
We tested the proposed adaptive control algorithm in simulation on a three tank system. This system consists of three water tanks that are mutually connected in series with narrow pipes that are attached to the tanks at their bottom and whose cross section can be controlled by valves. Water can be directly pumped from a water reservoir into the two outer tanks, but not into the tank in the middle. One of the outer tanks has a small opening at the bottom through which the water is allowed to leak out into the water reservoir. We assume that all three tanks have the same cross section that we denote by $S$. In addition, we assume that the cross sections of the connections between the tanks and the water outlet has the area given by $\gamma_i S_c, i =1,2,3$, where    $S_c$ is a constant term and $\gamma_i$ are the time-varying parameters whose values are defined by the positions of the corresponding valves. We further denote the water levels in the three tanks by $h_i, i =1,2,3$ and the input water flows into the tanks 1 and 3 by $q_1$ and $q_2$. Fig. \ref{Fig: tanks} shows the physical organization of the described three tank system. 

\begin{figure}[!hbt]
	\centerline{
		\includegraphics[width=\columnwidth]{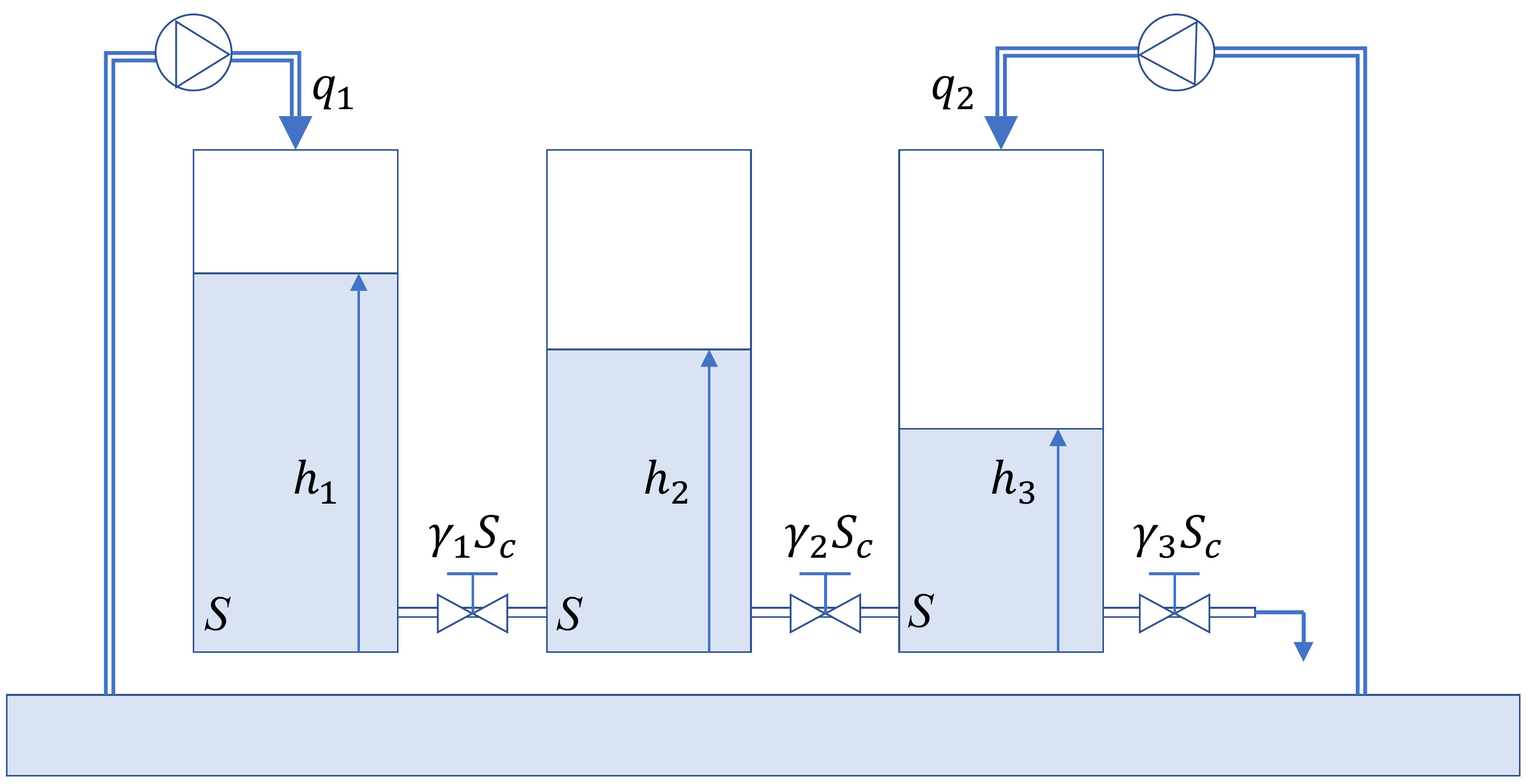}}
	\caption{The three tank system}
	\label{Fig: tanks}
\end{figure}

If we denote the Earth's gravity acceleration constant by $g$, then the dynamic equations that describe the evolution of the water levels in the three tanks are:
\beq\label{eq:tanks}
\begin{array}{ccl}
\frac{dh_1}{dt}&\!\! = \!\!& \frac{q_1 -\gamma_1S_c sgn(h_1\! -\! h_2)\sqrt{2g(h_1\! -\! h_2)}}{S}\\
\frac{dh_2}{dt}&\!\! =\!\! & \frac{\gamma_1S_c sgn(h_1\! -\! h_2)\sqrt{2g(h_1 \!- \!h_2)} -\gamma_2S_c sgn(h_2\! -\! h_3)\sqrt{2g(h_2\! -\! h_2)}}{S}\\
\frac{dh_3}{dt}&\!\! = \!\!& \frac{q_2 -\gamma_2S_c sgn(h_2\! -\! h_3)\sqrt{2g(h_2\! -\! h_2)}-\gamma_3S_c\sqrt{2gh_3}}{S}
\end{array}
\eeq  

In simulations, we modify the values of the parameters $\gamma_1$ and $\gamma_3$ over time as shown in Fig. \ref{Fig: parameters}. Numerical values for all other three tank model parameters are listed in Table \ref{T: tank_parameters}.

\begin{table}[ht]
\caption{Numerical values of the three tank model parameters.} 
\centering  
\begin{tabular}{ccc} 
\hline\hline                        
$S~[cm^2]$& $S_c~[cm^2]$& $\gamma_2$ \\ [0.5ex] 
\hline                  
375 & 3.42 & 0.5 \\ 
\hline 
\end{tabular}
\label{T: tank_parameters} 
\end{table}

\begin{figure}[!hbt]
	\centerline{
		\includegraphics[width=\columnwidth]{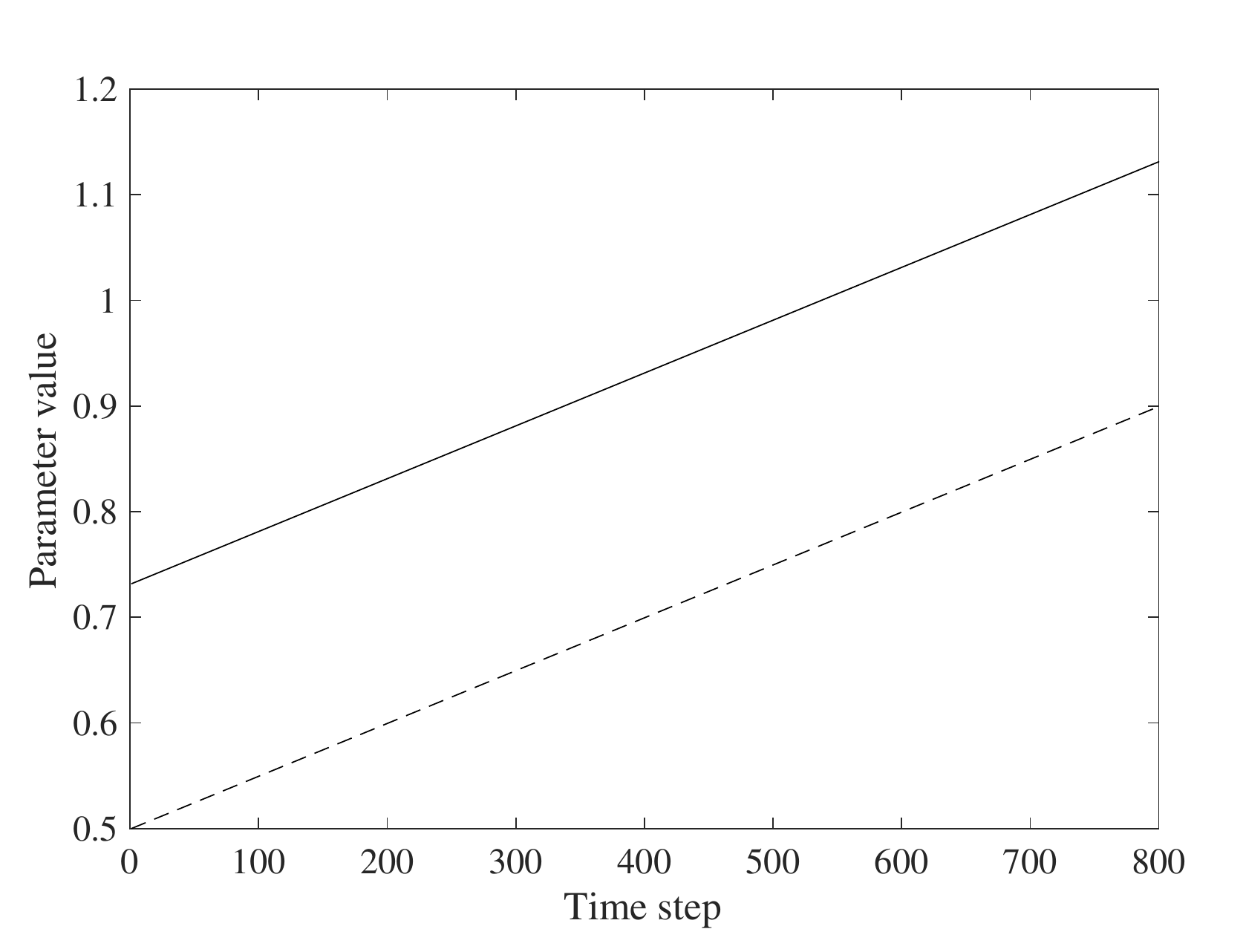}}
	\caption{Time variation of the parameters $\gamma_1$ (dashed line) and $\gamma_3$ (solid line)}
	\label{Fig: parameters}
\end{figure}

We regulate the tank water levels around a steady state that is defined by the water levels $h_1 = 8\,$cm, $h_2 = 7\,$cm and $h_3 = 6\,$cm. Therefore, the simulations are done with the linearization of the system \eqref{eq:tanks} around these steady state values, where the plant outputs are the differences of the tank water levels and the steady state levels and the control inputs are the differences of the two water flows with respect to the steady state water flows. System is regulated with a sampling time of $0.16~s$. 

The described system has $2$ inputs and $3$ outputs (i.e. $n_u = 2$ and $n_y = 3$). We consider a finite impulse response model that uses $12$ coefficients to describe the influence of each input to each output (i.e. $m = 24$). The control objective is to regulate the system such that the water level in  tank 2 (i.e. $h_2$) follows a given reference profile and satisfy the input and output constraints. The constraints are selected such that the rate and amplitude of both control inputs are limited, that the water level of the first tank stays below $12\,$cm, that the level of the second tank remains below the level of the first tank and the level of the third tank remains below the level of the second tank and finally that the level of the third tank remains above $0\,$cm. These input and output constraints yeald the following values for the matrices and vectors in \eqref{Eq:IO_constraints}:
\[
\begin{aligned}
C_u = C_{\Delta u} &= \left[\ba{cc} 1 & 0\\0 & 1\\-1 & 0\\ 0 & -1 \ea  \right], &g_u=&\left[\ba{c}9\\9\\9\\9  \ea \right], g_{\Delta u}= \left[\ba{c}4\\4\\4\\4  \ea \right]\\
C_y&=\left[\ba{ccc} 1&0&0\\0&0&-1\\-1&1&0\\0&-1&1 \ea \right], &p=& \left[\ba{c}5\\6\\1\\1  \ea \right].
\end{aligned}
\]

The initial feasible parameter set $\mathcal{F}_0$ and the set of constrains on the model parameter's rate of change $\mathcal{D}$ (see \eqref{eq:change_bounds}) have been defined by choosing identical box constraints on the impulse response coefficients for each input-output pair. The physics of the considered plant defines the lower bound on each of the impulse response coefficients to be zero. The upper bounds on the impulse response coefficients are defined by using an exponentially decaying curve that over bounds the impulse response coefficients at each time step. The bounds on the rate of change of the impulse response coefficients is also defined by exponentially decreasing bounds, with the difference that a constant bound is assumed for the couple of first coefficients. The bounds on the impulse response coefficient magnitude and rate of change are sown in Fig. \ref{Fig: bounds}. Numerical values of other tuning parameters of the proposed adaptive MPC controller are listed in Table \ref{T: controller_parameters}. In simulations, additive noise uniformly distributed in the range defined by the bounds in Table \ref{T: controller_parameters} was used.

\begin{figure}[!hbt]
	\centerline{
		\includegraphics[width=11 cm]{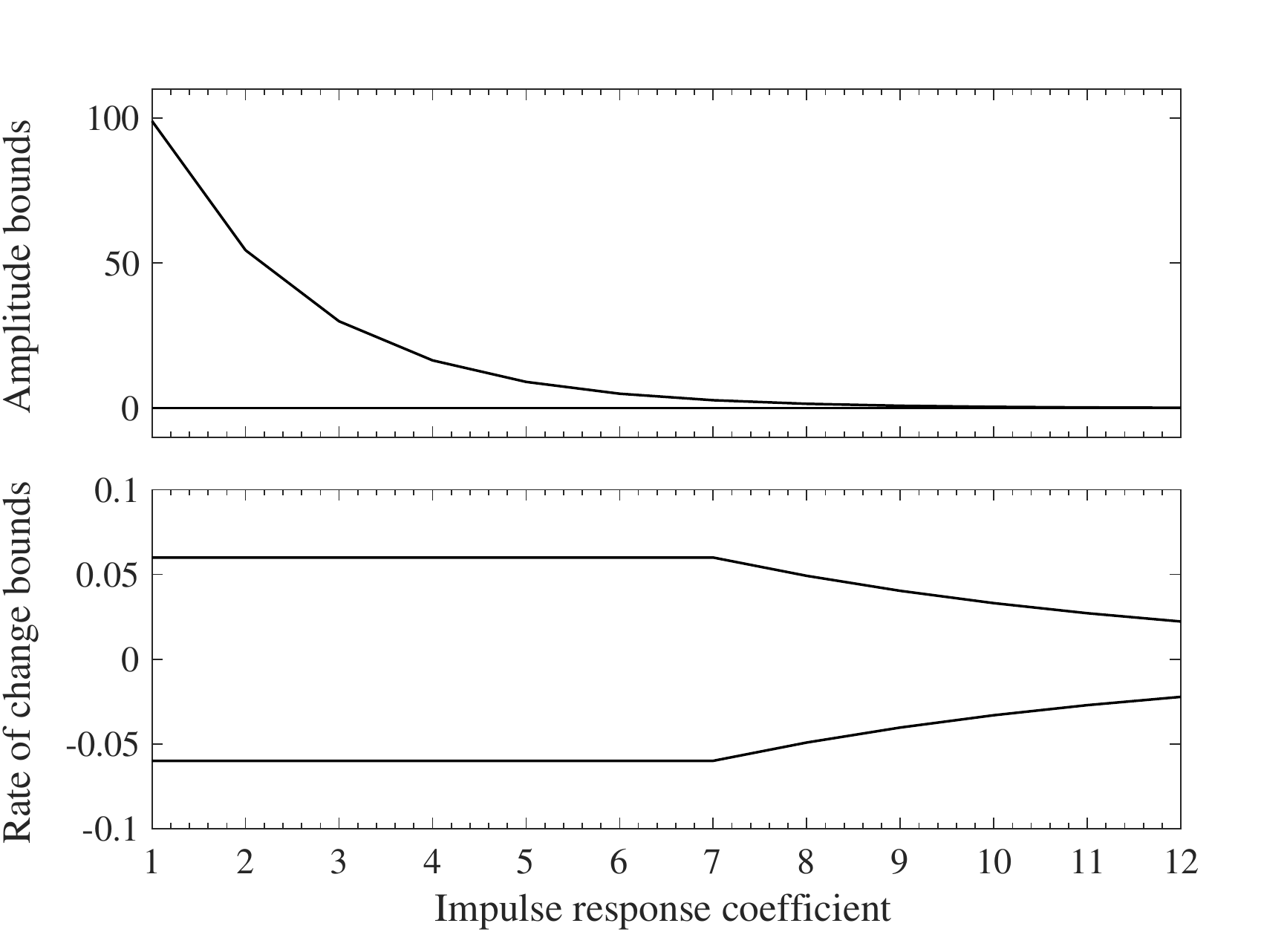}}
	\caption{Bounds on the impulse response coefficient amplitudes (upper plot) and rates of change (lower plot)}\label{Fig: bounds}
\end{figure}

\begin{table}[ht]
	\caption{Numerical values of the controller tuning parameters.} 
	\centering  
	\begin{adjustbox}{max width=1\columnwidth}
	\begin{tabular}{ccccccc} 

		\hline\hline                        
		$\epsilon_d$ & $\epsilon_v$ & $Q$& $R$& $S$ & $N$ & $M$ \\ [0.5ex] 
		\hline                  

		$\left[\ba{c}0.1\\0.1\\0.1\ea \right]$&$\left[\ba{c}0.1\\0.1\\0.1\ea \right]$& $\left[\ba{ccc} 0 & 0 & 0\\0 & 1 & 0\\0 & 0 & 0\ea  \right]$ & $\left[\ba{cc} 0.5 & 0\\0 & 0.5\ea\right]$ & $\left[\ba{cc} 0 & 0\\0 & 0\ea\right]$ & $22$ & $100$ \\ 
	
		\hline 

\end{tabular}
	\label{T: controller_parameters} 
	\end{adjustbox}
\end{table}

Resulting tank water levels and control inputs are shown in Fig. \ref{Fig: results} and Fig. \ref{Fig: inputs} respectively. In addition to the resulting plant outputs, Fig. \ref{Fig: results} also shows the upper and the lower bounds for each of the three outputs with respect to the feasible parameter set at each time step. As can be seen, the output constraints are maintained for the whole range of uncertainty, which results in robust constraint satisfaction.

\begin{figure}[!htb]
	\centerline{
		\includegraphics[width=11 cm]{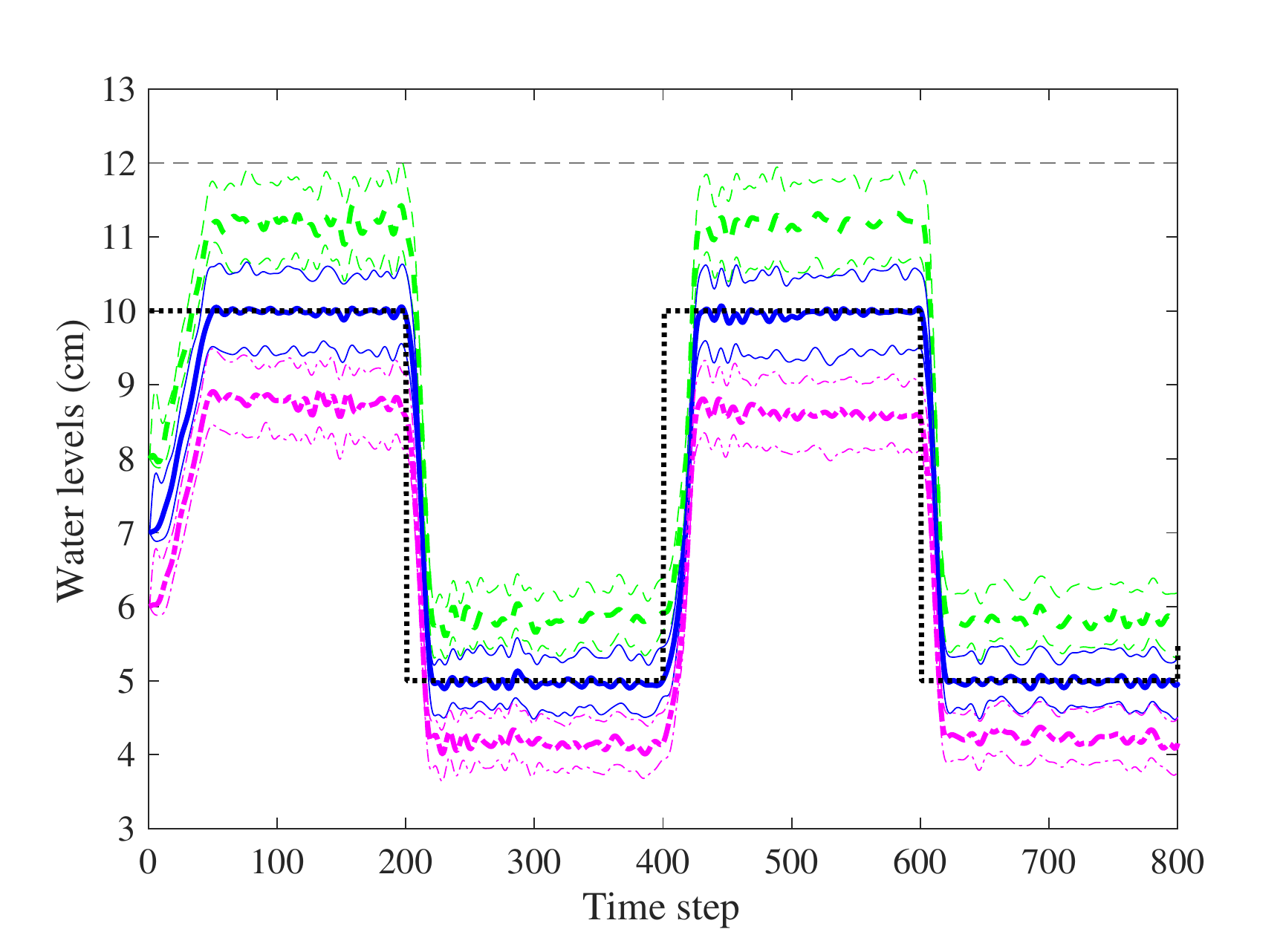}}
	\caption{Resulting tank water levels obtained when the proposed adaptive MPC algorithm is used (thick lines) for the first (green dashed), second (blue solid) and the third (magenta dash-dot) tank, compared with the reference for the water level in the second tank (thick dotted line). In addition to the simulated tank water levels, the uncertainty intervals calculated based on the feasible parameter set are also shown (thin lines), as well as the constraint of 12~cm (black dashed line). }
	\label{Fig: results}
\end{figure}

\begin{figure}[!hbt]
	\centerline{
		\includegraphics[width=11 cm]{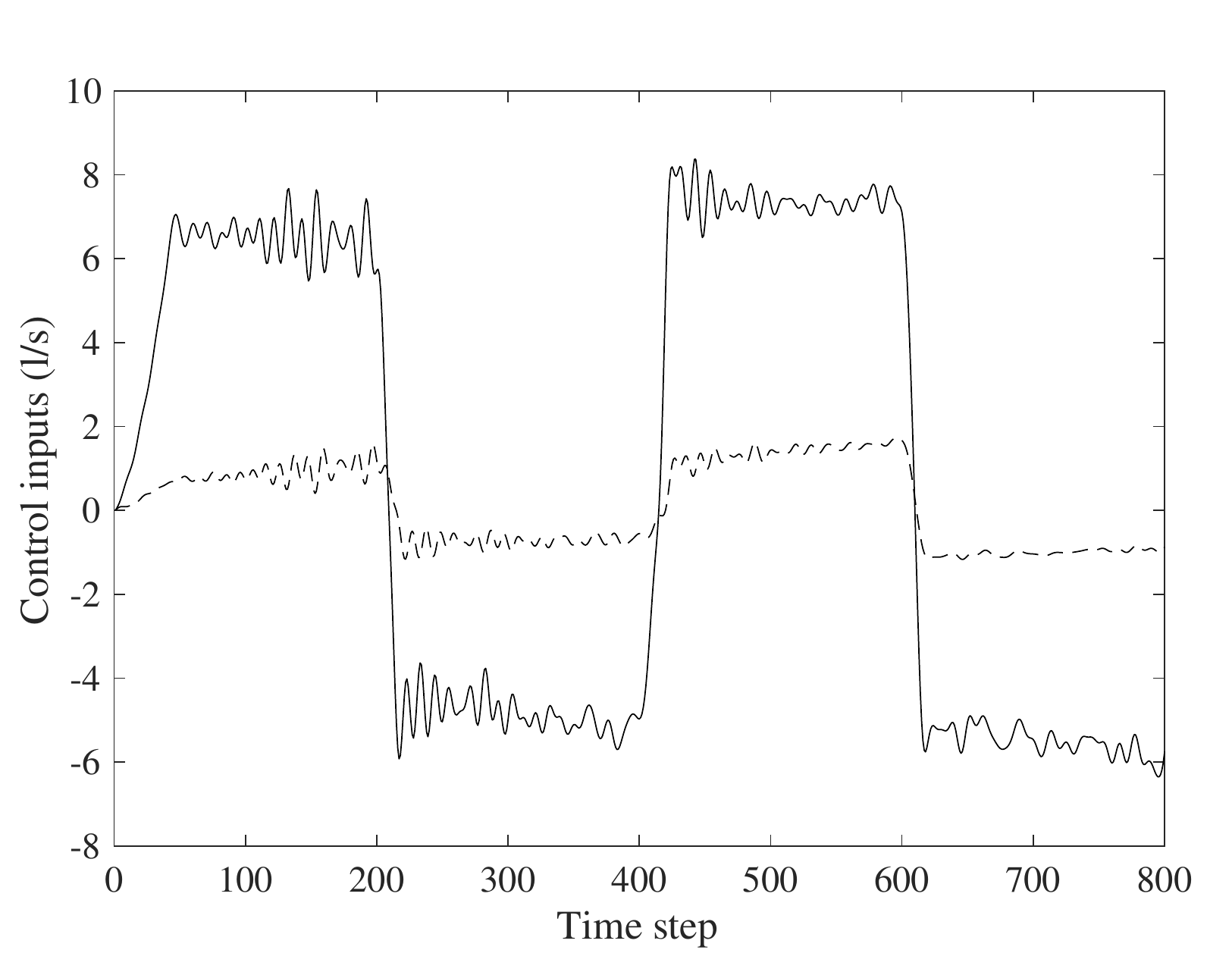}}
	\caption{Control inputs $u_1$ (dashed line) and $u_2$ (solid) obtained in the simulation.}
	\label{Fig: inputs}
\end{figure}

To illustrate the effectiveness of the proposed adaptive control scheme, we compared its performance with the performance of the identical MPC controller that uses least squares with forgetting. For the simulations a forgetting factor of $0.9$ was used. Both controllers used the same initial guess for the plant parameters, and the controller that uses least squares implements a soft enforcement of the output constraints, as there are no recursive feasibility guarantees in this case. The tank water levels obtained with this controller are shown in Fig. \ref{Fig: Least_squares1}. As can be seen, the use of this controller results in output constraint violation, which is shown in greater detail in Fig. \ref{Fig: Least_squares2}. The adaptive controller with least squares is much less conservative as it does not take the uncertainty into account. On the other hand, although more cautious, the proposed adaptive MPC algorithm for time varying systems is capable of satisfying output constraints and guarantees recursive feasibility. 

\begin{figure}[!hbt]
	\centerline{
		\includegraphics[width=11 cm]{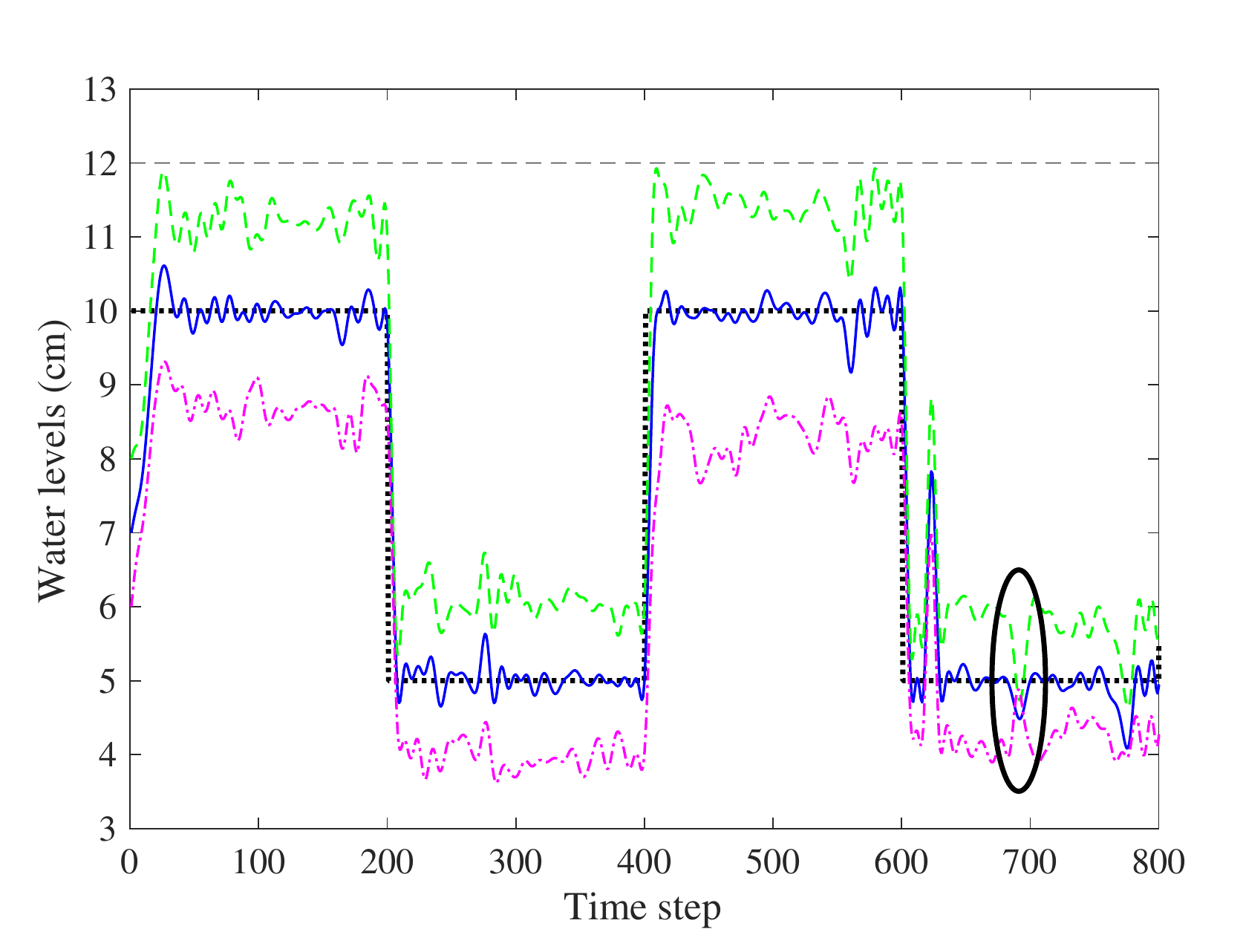}}
	\caption{Tank water levels obtained when the adaptive MPC algorithm based on recursive least squares with forgetting is used for the first (green dashed line), second (blue solid) and  third (magenta dash-dot) tank, compared with the reference signal for the water level in the second tank (thick dotted line). An example of output constraint violation is marked with a black ellipsoid (enlarged in Fig. \ref{Fig: Least_squares2}).}
	\label{Fig: Least_squares1}
\end{figure}

\begin{figure}[!hbt]
	\centerline{
		\includegraphics[width=11 cm]{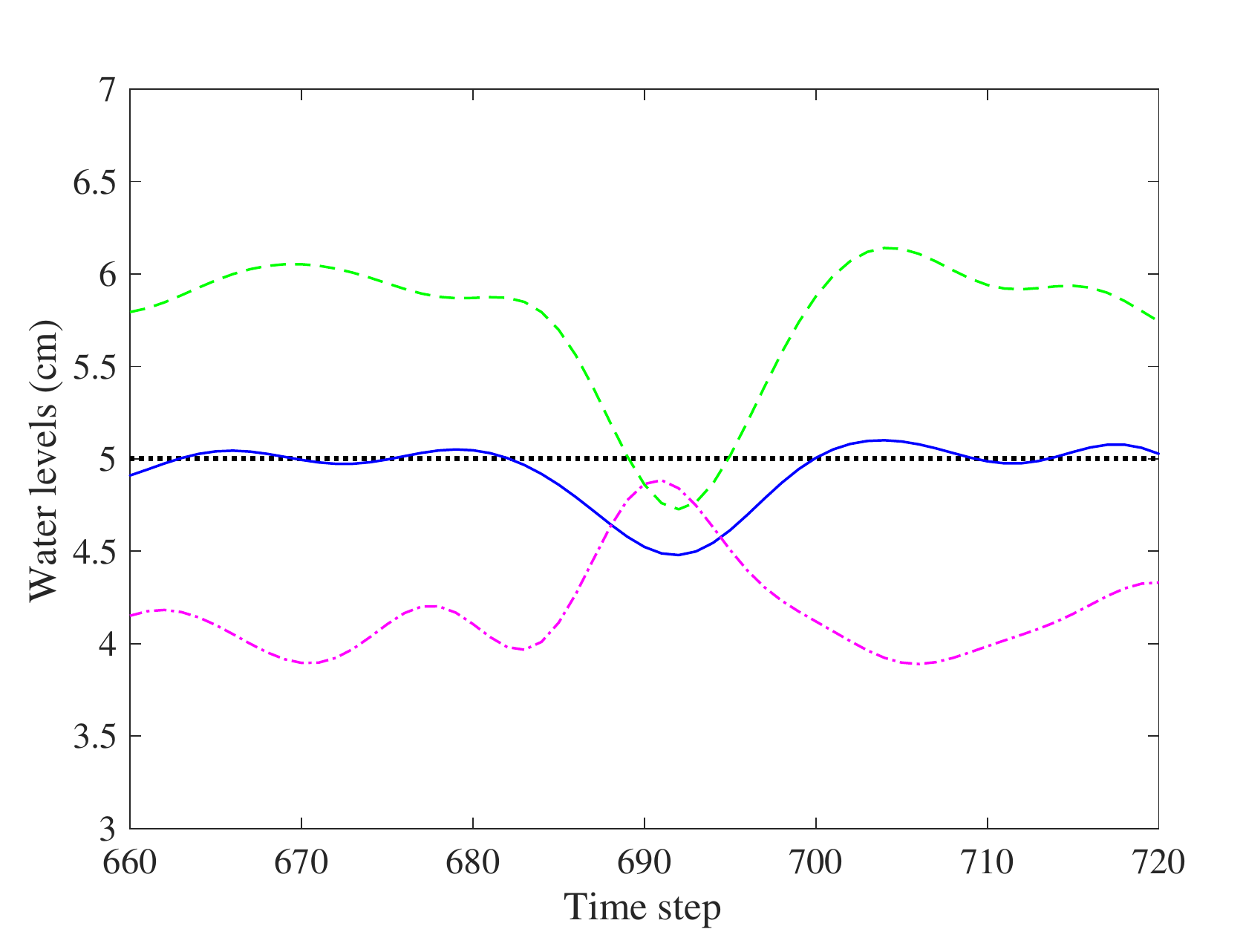}}
	\caption{Zoom-in of the constraint violation that occurs when the adaptive MPC based on recursive least squares with forgetting is used, corresponding to the circled area in Fig. \ref{Fig: Least_squares1}. According to the output constraints, the three water levels should never cross.}
	\label{Fig: Least_squares2}
\end{figure}


\section*{Appendix}\small
\emph{Proof of Lemma \ref{L:nonempty_membership}}. We use induction to prove the clain of the Lemma. At time step $t=0$, from the step 1) of Algorithm \ref{A: Recusive_identification}, it holds that $\mathcal{F}(0)=\Omega$ and from Assumption \ref{A:system_class}, it then follows that $\mathcal{F}(0)\neq \emptyset$ and that $H(0)\in \mathcal{F}(0)$. Let us now, for the sake of the inductive argument, assume that at some time step $t\geq 0$, it holds that $H(t)\in \mathcal{F}(t)$. We shall show, that it than follows that $H(t+1)\in \mathcal{F}(t+1)$. To this end, we define matrices $A_j'(t)\in \bR^{r_{j0}\times m},\, A''_j(t)\in \bR^{(r_j(t)-r_{j0})\times m}$ and vectors $b_j'(t)\in \bR^{r_{j0}},\, b''_j(t)\in \bR^{r_j(t)-r_{j0}}, j=1,\hdots, n_y$, as:
\[
b'(t)=\left[\ba{c}b_{j1}\\\vdots \\ b_{jr_{j0}} \ea \right],\, b''(t) = \left[\ba{c}b_{jr_{j0}+3}\\\vdots\\b_{jr_j(t)}\ea \right],\, A'_j(t)\!=\!\left[\! \ba{c}a_{j1}\\\vdots \\ a_{jr_{j0}} \ea\!\right],\,
\]
\[
A''_j(t)\!=\!\left[\!\ba{c}a_{jr_{j0}+3}\\\vdots\\a_{jr_j(t)}\ea \!\right]\!=\!\left[\!\ba{c} -\varphi^T\left(t\!-\!\frac{r_j(t)\!-\!r_{j0}}{2}\right)\\\varphi^T\left(t\!-\!\frac{r_j(t)\!-\!r_{j0}}{2}\right)\\ \vdots \\-\varphi^T(t)\\\varphi^T(t)\ea \!\right].
\]
Note that $A_j(t)=\left[\ba{c}A'_j(t)\\A''_j(t)\ea\right]$ and $b_j(t)=\left[\ba{c}b'_j(t)\\b''_j(t)\ea\right]$. From Assumption \ref{A:system_class}, it holds that:
\beq\label{eq:P1_ineq1}
A'_j(t)H_j(t+1)\leq b'_j(t), j=1,\hdots, n_y.
\eeq
In addition, we note that from the inductive assumptions, it holds that $A''_j(t)H_j(t)\leq b''_j(t), j=1,\hdots, n_y$. Therefore, it than also holds that:
\[
A''_j(t)H_j(t+1)\leq b_j''(t)+e_j(t)
\]
where $e_j(t)\in \bR^{r_j(t)-r_{j0}}$, $e_j(t)=A''_j(t)\left(H_j(t+1)-H_j(t) \right)$, $j=1,\hdots, n_y$. From the definition of $A''_j(t)$ (note that this matrix is exclusively formed from the past regressor vectors), and the definition of   $\overline{\vartheta}(t)$ and $\underline{\vartheta}(t)$ in \eqref{eq: time_varying_bounds}, we note that the vectors $e_j(t), j=1,\hdots, n_y$ are bounded such that it holds:
\[
e_j(t)\leq \overline{e}_j = \left[\ba{c}   -\underline{\vartheta}\left(t-\frac{r_j(t)-r_{j0}}{2}\right)\\ \overline{\vartheta}\left(t-\frac{r_j(t)-r_{j0}}{2}\right) \\ \vdots \\  -\underline{\vartheta}(t)\\ \overline{\vartheta}(t)   \ea\right].
\]  
Therefore, it holds that: 
\beq\label{eq:P1_ineq2}
A''_j(t) H_j(t+1)\leq b_j''(t+1)+\overline{e}_j, j=1,\hdots, n_y.
\eeq
Moreover, from Assumption \ref{A:disturbance}, it follows that the following two inequalities have to be satisfied:
\beq\label{eq:P1_ineq3}
\ba{ll}
-\varphi(t\!+\!1)H_j(t\!+\!1)&\!\leq\! -\tilde{y}_j(t\!+\!1)\!+\!\epsilon_{d_j}\!+\!\epsilon_{v_j} \\
-\varphi(t\!+\!1)H_j(t\!+\!1)&\! \leq \!\tilde{y}_j(t\!+\!1)\!+\!\epsilon_{d_j}\!+\!\epsilon_{v_j}
\ea,  j\!=\!1,\hdots, n_y.
\eeq
Based on \eqref{eq:P1_ineq1}, \eqref{eq:P1_ineq2} and \eqref{eq:P1_ineq3}, it holds that:
\[
A_j^\dag(t+1)H_j(t+1)\leq b_j^\dag(t+1),\, j=1,\hdots, n_y,
\] 
where
\[
A_j^\dag(t\!+\!1)\!=\!\left[\ba{c}A'_j(t)\\A''_j(t)\\-\varphi(t\!+\!1)\\\varphi(t\!+\!1) \ea\right],
\]
and
\[
\, b_j^\dag(t+1)\!=\!\left[\ba{c}b'_j(t)\\b''_j(t\!+\!1)+\overline{e}_j\\-\tilde{y}_j(t\!+\!1)+\epsilon_{d_j}\!+\!\epsilon_{v_j}\\\tilde{y}_j(t\!+\!1)\!+\!\epsilon_{d_j}\!+\!\epsilon_{v_j} \ea\right],
\]
are the matrices that would be obtained after running the step 4) of Algorithm \ref{A: Recusive_identification} at time $t+1$ (i.e. before removing any rows from the matrices and vectors in order to keep their dimensions bounded). Therefore, the set $\mathcal{F}^\dag(t+1)= \left\{ H \in \bR^{n_y \times m}:A_j^\dag(t+1)H_j\leq b_j^\dag(t+1)\right\}$ is a nonempty set that is guaranteed to contain $H(t+1)$, i.e. $H(t+1)\in \mathcal{F}^{\dag}(t+1)$. Set $\mathcal{F}^{\dag}(t+1)$ represents the updated feasible parameter set before possible removal of any inequalities in order to bound the complexity of its description. The set $\mathcal{F}(t+1)$ is obtained by either taking the set $\mathcal{F}^\dag(t+1)$ as it is (i.e. when $r_j(t)\leq M+r_{j0},\forall j=1,\hdots, n_y$), or by removing several inequalities that constitute it (see step 5) of Algorithm \ref{A: Recusive_identification}). Therefore it holds that $\mathcal{F}^\dag(t+1)\subseteq \mathcal{F}(t+1)$, and hence it holds that $H(t+1)\in \mathcal{F}(t+1)$, which means that $ \mathcal{F}(t+1)\neq \emptyset$. By invoking the argument of mathematical induction, it then holds that $H(t) \in \mathcal{F}(t),\forall t\geq 0$, which completes the proof.$\hfill\blacksquare$\\

\emph{Proof of Lemma \ref{L:predicted_FPS}}. We first note that, from the definition of $\mathcal{F}(t+1|t)$ (see \eqref{eq:predicted_FPS},\eqref{Eq:predicted_matrices1} and \eqref{Eq:predicted_matrices2}), and the way Algorithm \ref{A: Recusive_identification} works, it holds that:
\[	A_j(t\!+\!1)\! =\!\!\left[\!\!\ba{c}A_j(t\!+\!1|t)\\-\varphi(t+1)\\\varphi(t+1) \ea\!\!\right],  b_j(t\!+\!1)\!=\!\!\left[\!\!\ba{c}b_j(t\!+\!1|t)\\-\tilde{y}_j(t\!+\!1)+\epsilon_{d_j}\!+\!\epsilon_{v_j}\\\tilde{y}_j(t\!+\!1)\!+\!\epsilon_{d_j}\!+\!\epsilon_{v_j}  \ea\!\!\right]\!\!.
\]
Matrices $A_j(k|t+1)$ and vectors $b_j(k|t+1), j=1,\hdots, n_y$ are then, by construction, formed from the matrices $A_j(t+1)$ and $b_j(t+1)$. Therefore we have that, for $j=1,\hdots, n_y$ and $k\in[t+2,t+N]$, it holds:
\[
A_j(k|t+1) = \left[\ba{c}A_j(k|t)\\-\varphi(t+1)\\\varphi(t+1)\ea \right],
\]
and
\[
b_j(k|t+1) = \left[\ba{c}b_j(k|t)\\ -\tilde{y}_j(t+1)+\epsilon_{d_j}+\epsilon_{v_j}\\\tilde{y}_j(t+1)+\epsilon_{d_j}+\epsilon_{v_j}   \ea \right]
\]
As it holds that $\mathcal{F}(k|t) = \left\{ H \in \bR^{n_y \times m}:A_j(k|t)H_j\leq b_j(k|t)\right\}$, and $\mathcal{F}(k|t+1) =$ 

$ \left\{ H \in \bR^{n_y \times m}:A_j(k|t+1)H_j\leq b_j(k|t+1)\right\}$, $\forall k\in [t+2,t+N-1]$, it holds that each of the sets $\mathcal{F}(k|t+1), k\in[t+1,t+N-1]$ is formed by the same inequalities as the set $\mathcal{F}(k|t)$ and that it has two additional inequalities defined by the regressor vector and output measurement at time step $t+1$. Therefore, it holds that $\mathcal{F}(k|t+1)\subseteq \mathcal{F}(k|t), k\in [t+2,t+N-1]$. In addition, we note that $\mathcal{F}(t+N|t)=\Omega$ and that for $j=1,\hdots, n_y$, it holds that:
\[
A_j(t+N|t+1)=\left[\ba{c}A_{j0}\\A_j'  \ea \right], \, b_j(t+N|t+1)=\left[\ba{c}b_{j0}\\b_j'  \ea \right],
\]
where the matrices $A_j'$ and the vectors $b_j', j=1,\hdots, n_y$ are obtained by using the  rules for generating the predicted matrices $A_j(k|t)$ and vectors $b_j(k|t)$ in ,\eqref{Eq:predicted_matrices1} and \eqref{Eq:predicted_matrices2}. Therefore, from the definition of $\mathcal{F}(t+N|t+1)$ (see e.g. \eqref{eq:predicted_FPS}) and the definition of the set $\Omega$ in \eqref{Eq:system_set}, it holds that $\mathcal{F}(t+N|t+1)\subseteq \mathcal{F}(t+N|t)$. Hence, it holds that $\mathcal{F}(k|t+1)\subseteq \mathcal{F}(k|t), k\in [t+2,t+N]$, which completes the proof.
$\hfill\blacksquare$\\

\emph{Proof of Theorem \ref{L:reqursive_feasibility_tv}}. We first show that the FHOCP \eqref{Eq:FHOCP} is recursively feasible. To this end, we use induction. The problem \eqref{Eq:FHOCP} is feasible for $t=0$ by assumption. Let us assume that the problem \eqref{Eq:FHOCP} is feasible at a generic time step $t$ and let the optimal control sequence be $U^*(t)=[u^*(t|t),\ldots,u^*(t+N-1|t)]$, and its corresponding sequence of predicted regressor vectors be $\varphi^*(k|t),\, k=t+1,\hdots,t+N$. Then, a possible feasible control sequence at $t+1$ is $U(t+1)=[u^*(t+1|t),\ldots,u^*(t+N-1|t),u^*(t+N-1|t)]$. This sequence satisfies constraints \eqref{Eq:input_rate} and \eqref{Eq:end_constraint}. In addition, we note that the predicted regressor vectors $\varphi(k|t+1), \, k=t+2,\hdots, t+N+1$ that correspond to the input sequence $U(t+1)$, by construction satisfy the equalities $\varphi(k|t+1)=\varphi^*(k|t)$, for $k\in [t+2,t+N]$ and that from \eqref{Eq:end_constraint} it follows that $\varphi(t+N+1|t+1)=\varphi^*(t+N|t)$. Moreover, we note that from Lemma \ref{L:predicted_FPS}, it holds that $\mathcal{F}(k|t+1)\subseteq \mathcal{F}(k|t), \forall k\in[t+1,t+N]$ In addition, we note that $\mathcal{F}(t+N+1|t+1)=\mathcal{F}(t+N|t)=\Omega$. Based on this, the sequence of inputs $U(t+1)$ satisfies the output constraints \eqref{Eq:output}, which means that the constraints \eqref{Eq:relaxed_constraints1} are feasible and hence the FHOCP \eqref{Eq:FHOCP} has a feasible solution. Repeating this argumentation for all $t>0$, it can be concluded that the FHOCP \eqref{Eq:FHOCP} remains feasible $\forall t>0$. From this and Lemma \ref{L:nonempty_membership}, the other claim of the Theorem follows directly.$\hfill \blacksquare$

\normalsize
\bibliographystyle{plain}

\end{document}